\documentclass{article}

\usepackage[a4paper, total={6in, 8in}]{geometry}

\usepackage[utf8]{inputenc} 
\usepackage[T1]{fontenc}    
\usepackage{hyperref}       
\usepackage{url}            
\usepackage{booktabs}       
\usepackage{amsfonts}       
\usepackage{nicefrac}       
\usepackage{microtype}      
\usepackage{xcolor} 
\newcommand{\tva}{\mathbf{TV}_a}
\newcommand{\adv}{\mathbf{Adv}}
\newcommand{\tv}{\mathbf{TV}}
\newcommand{\cN}{\mathcal{N}}
\newcommand{\nrm}[2]{\lVert #1 \rVert_{#2}}
\newcommand{\pfix}[2]{{#1}_{\leq #2}}
\newcommand{\set}[1]{\{#1\}}
\usepackage{microtype}
\usepackage{graphicx}
\usepackage{subfigure}
\usepackage{booktabs} 
\usepackage{natbib}

\usepackage{hyperref}

\usepackage{amsmath}
\usepackage{amssymb}
\usepackage{amsthm}
\usepackage{appendix}
\usepackage{bbm}
\usepackage{mathtools}
\newcounter{thm}

\newtheorem{theorem}[thm]{Theorem}
\newtheorem{lemma}[thm]{Lemma}

\newtheorem{definition}[thm]{Definition}
\newtheorem{remark}[thm]{Remark}

\newcommand\numberthis{\addtocounter{equation}{1}\tag{\theequation}}
\newcommand{\prob}{\mathbb{P}}

\newcommand{\R}{\mathbb{R}}
\newcommand{\E}{\mathbb{E}}

\newcommand{\calM}{\mathcal{M}}
\newcommand{\calA}{\mathcal{A}}

\newcommand{\Supp}{\mathsf{Supp}}

\DeclarePairedDelimiterX{\infdivx}[2]{(}{)}{%
  #1\;\delimsize\|\;#2%
}
\newcommand{\rdp}{D_\alpha\infdivx}
\newcommand{\kl}{\text{KL}\infdivx}
\usepackage{xcolor}

\newcommand{\erf}{\mathrm{erf}}

\date{}
\title{\textbf{Optimal Membership Inference Bounds for Adaptive Composition of Sampled Gaussian Mechanisms}}

\author{

\text{Saeed Mahloujifar}\footnote{Equal contribution}\\
\it{Princeton University}\\
\text{sfar@princeton.edu}
\and
\text{Alexandre Sablayrolles}$^\ensuremath{\ast}$\\
\it{Meta AI}\\
\text{asablayrolles@fb.com }
\and
\text{Graham Cormode}\\
\it{Meta AI}\\
\text{gcormode@fb.com}
\and
\text{Somesh Jha}\\
\it{University of Wisconsin}\\
\text{jha@cs.wisc.edu}}

\begin{document}

\maketitle

\begin{abstract}
Given a trained model and a data sample, membership-inference (MI) attacks predict whether the sample was in the model's training set.
A common countermeasure against MI attacks is to utilize differential privacy (DP) during model training to mask the presence of individual examples.
While this use of DP is a principled approach to limit the efficacy of MI attacks, 
there is a gap between the bounds provided by DP and the empirical performance of MI attacks.
In this paper, we derive bounds for the \textit{advantage} of an adversary mounting a MI attack, 
and demonstrate tightness for the widely-used Gaussian mechanism.
We further show bounds on the \textit{confidence} of MI attacks. Our bounds are much stronger than those obtained by DP analysis. For example, analyzing a setting of DP-SGD with $\epsilon=4$ would obtain an upper bound on the advantage of $\approx0.36$ based on our analyses, while getting bound of $\approx 0.97$ using the analysis of previous work that convert $\epsilon$ to membership inference bounds.

Finally, using our analysis, we provide MI metrics for models trained on  CIFAR10 dataset. To the best of our knowledge, our analysis provides state-of-the-art membership inference bounds for the privacy.
\end{abstract}

\section{Introduction}

The recent success of machine learning models make them the go-to approach to solve a variety of problems, ranging from computer vision~\citep{krizhevsky2012imagenet} to NLP~\citep{sutskever2014sequence}, including applications to sensitive data such as health records or chatbots. 
Access to a trained machine learning model, through a black-box API or a white-box access to a published model, can leak traces of information~\citep{dwork2015robust} from the training data. Researchers have tried to measure this information leakage through metrics such as membership inference~\citep{shokri2017membership}.
Membership inference is the task of guessing, from a trained model, whether it includes a given sample or not.
This task is both interesting in its own right, as the participation of an individual in a data collection can be a sensitive information.
It also serves as the ``most significant bit'' of information: if membership inference fails, attacks revealing more information such as reconstruction attacks~\citep{fredrikson2014privacy,carlini2020extracting} will also fail. In other words, defending against membership inference attacks would also defend against attacks such as reconstruction attacks that aim at reconstructing training examples.

The standard approach to provably defeat these membership privacy attacks is differential privacy~\citep{dwork2006calibrating}. Differential privacy defines a class of training algorithms that respect a privacy budget $\epsilon$ and a probability of failure $\delta$.
These quantities quantify how much information about each individual training example is revealed by the output of the algorithm. Most algorithms obtaining differential privacy need to inject noise somewhere in their process. The amount of injected noise then creates a trade-off between privacy utility of the trained model. To measure the privacy of a given algorithms, researchers have developed advanced mathematical tools and notions such as Renyi differential privacy \cite{mironov2017renyi,abadi16deep} and advanced composition theorems \cite{dwork2010boosting, kairouz2015composition}. These tools allow us to calculate $(\epsilon,\delta)$ values for carefully designed algorithms.

Previous work has shown that any deferentially private algorithm will provably bound the accuracy of any membership inference adversary. Specifically, starting with a given $(\epsilon,\delta)$, \cite{humphries2020differentially} prove that any model trained with $(\epsilon,\delta)$ differential privacy will induce an upper bound on the accuracy of membership adversary and this upper bound depends only on $\epsilon$ and $\delta$. These upper bounds enables us to obtain provable defenses against membership inference attacks by using deferentially private algorithms. In fact, there is a large gap between upper bounds proved for the power of adversaries in performing membership inference, and the power of real adversaries that try to attack deferentially private models.  One hypothesis is that the membership inference bound obtained by differential privacy is in reality stronger than what we could prove. In particular, the process of obtaining differential privacy bounds and then converting those bounds to membership inference bounds could be sub-optimal . In this work we ask the following question:

\begin{quote}
\begin{center}
\it{
Can we develop tools to directly and optimally analyze the membership inference upper bounds for algorithms, without going through differential privacy?}
\end{center}
\end{quote}

\paragraph{Our contributions:} Our contributions in this work are as follows:
\begin{itemize}
    \item \textbf{Membership inference bounds for composition of sampled Gaussian mechanisms} Our main theorem bounds the membership inference advantage of any adversary the composition of an arbitrary set of sampled Gaussian mechanisms that could adaptive depend on each other. Specifically, for any adaptive series of sampled Gaussian mechanism $(M_1,\dots,M_T)$ where $M_i$ has sensitivity $1.0$ and standard deviation $\sigma_i$ and sub-sampling rate $q_i$, we show that the membership inference advantage of any adversary is bounded by the total variation distance between two mixture of Gaussians defined by $\sigma=(\sigma_1,\dots,\sigma_T)$ and $q=(q_1,\dots,q_T)$. This bound is optimal as it reflects the membership inference advantage for a real adversary on a particular series of Gaussian mechanisms. 
    \item We propose a numerical way to calculate the total variation distance between mixture of Gaussians. Our algorithm is computationally tractable and works in linear time with respect to number of mechanisms and is independent from the dimension. We use our numerical approach to obtain concrete bounds on membership inference for Gaussian mechanisms and compare our bounds to that of \cite{humphries2020differentially}. 
    \item Finally, to understand the practical implication of our bound for mainstream datasets, we use DP-SGD to train models on CIFAR10 and calculate the membership inference bounds using our techniques. Our approach allows us to achieve state-of-the-art provable membership inference privacy for any given accuracy. 
\end{itemize}

The most widely-used DP algorithm in machine learning applications is DP-SGD: it is a small change of the classical stochastic gradient descent algorithm that only requires to clip per-sample gradients, average them and add Gaussian noise. DP-SGD has been shown to be more accurate than other differentially private training algorithms in the case of linear and convex models~\citep{van2020trade}. 
Each iteration of DP-SGD is an instance of the sampled Gaussian mechanism, which chooses a fraction $q$ of a dataset and outputs a noisy sum of the desired quantity.
Our analysis mirrors the recent shift in the field of empirical membership inference, from advantage (or accuracy) metrics~\citep{shokri2017membership,yeom2018privacy,sablayrolles2019white} to precision/recall measures~\citep{watson2021importance,carlini2020extracting}.
Differential privacy guarantees are known to yield tight true positive and false positive rates~\citep{nasr2021adversary}. 
Our paper is the first to show equivalent results in the case of advantage.
Our analysis explains, from a theoretical point of view, why the precision and recall of membership inference attacks is stronger than the accuracy.

\begin{figure}\label{fig:eps_adv}
    \centering
    \includegraphics[width=0.69\textwidth]{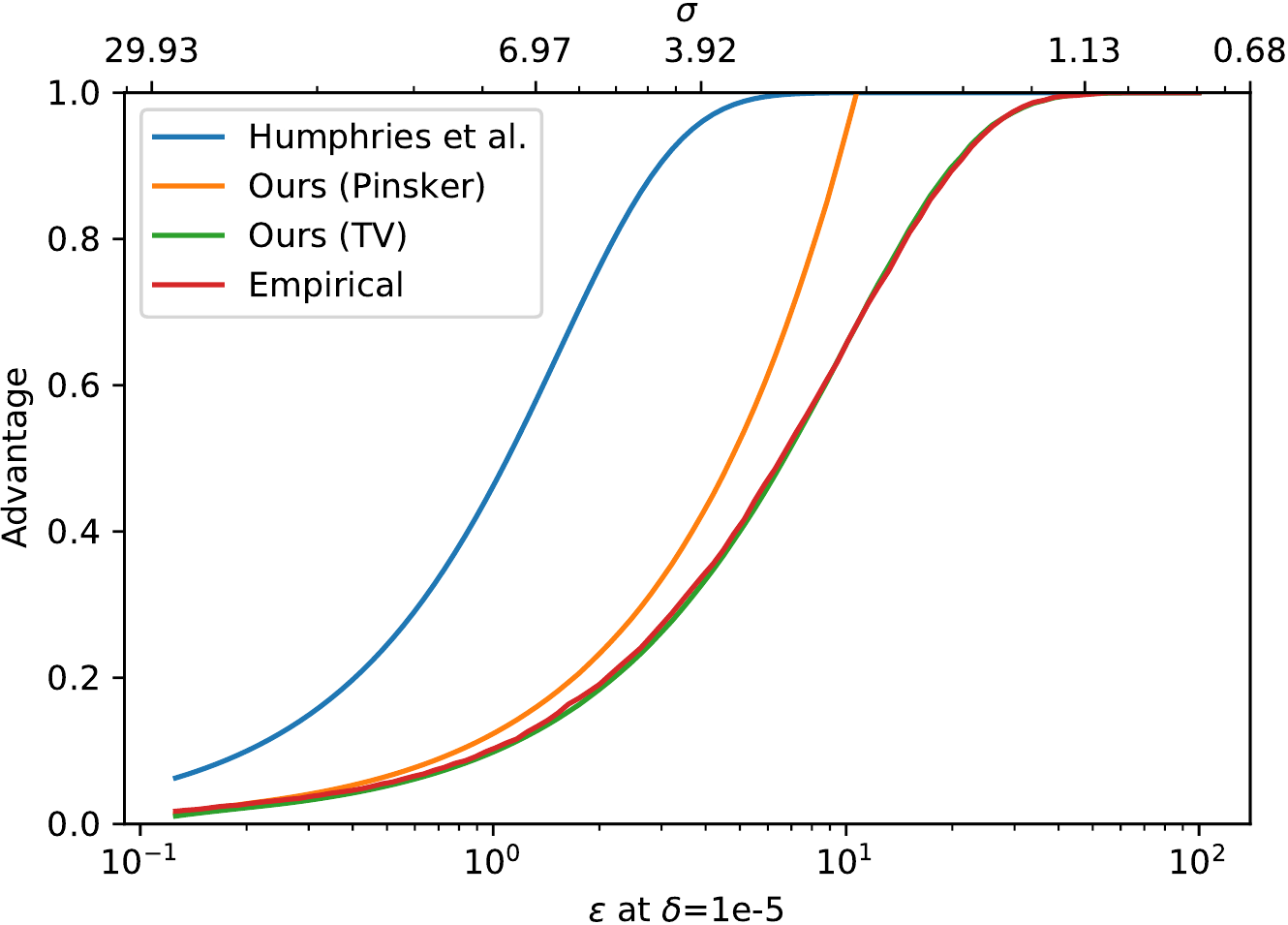}
    \caption{
    \label{fig:bounds}
    Bounds on membership advantage.
    Our bounds are tighter than \citet{humphries2020differentially} and match the empirical advantage.
    All bounds computed with a Gaussian mechanism with $C=1$ and varying $\sigma$.
    }
    
\end{figure}

\begin{figure}\label{fig:acc_adv}
    \centering
    \includegraphics[width=0.69\textwidth]{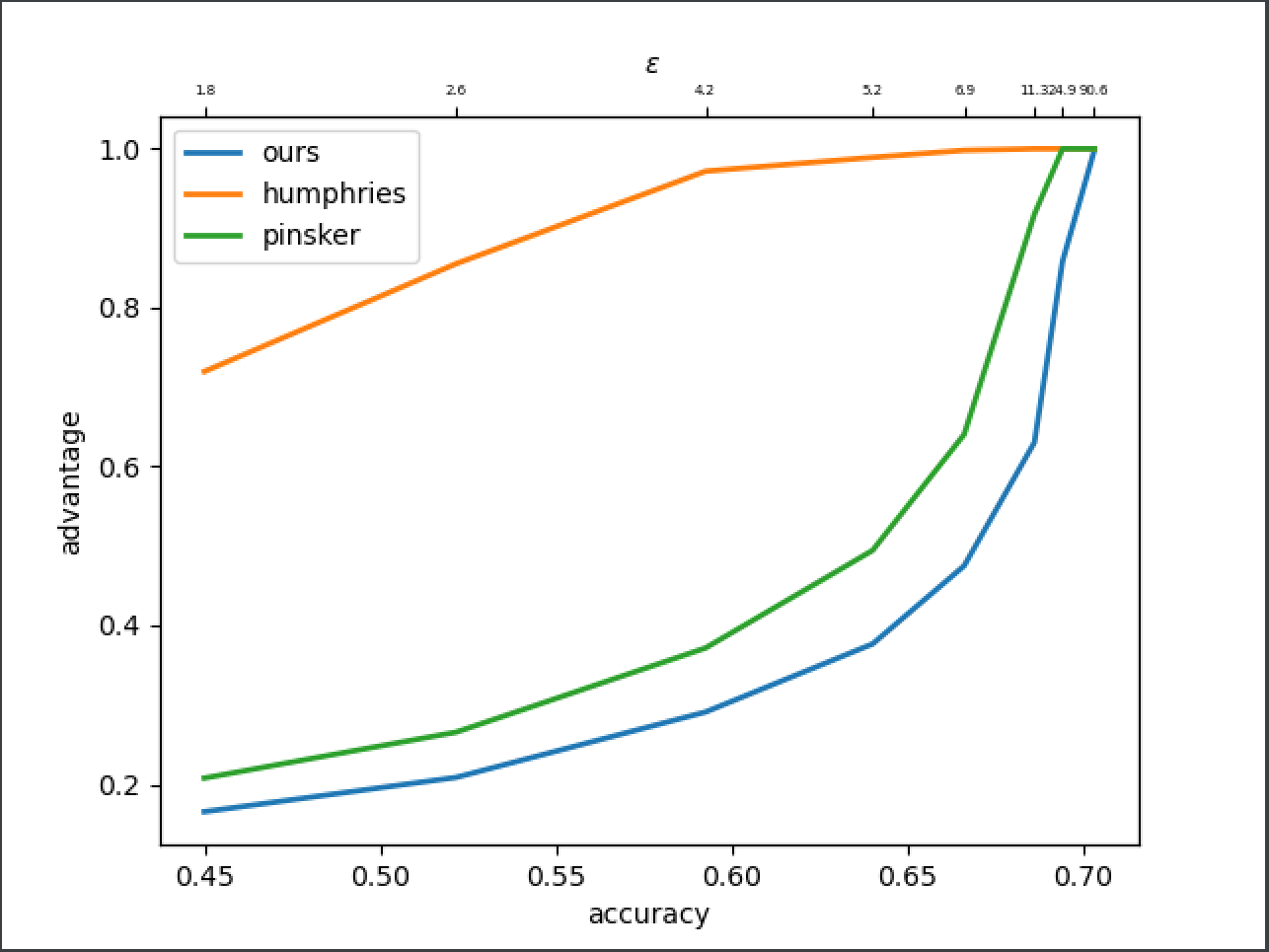}
    \caption{
    \label{fig:bounds}
    Experiments on CIFAR-10.
    This plot shows the effect of our improvement on the accuracy-privacy trade-off. All experiments are trained using dp-sgd for 50 epochs with sub-sample rate of $0.02$ and clipping threshold $10.0$, while varying the noise multiplier. 
    }
    
\end{figure}

\section{Background}
In this section we provide all the background information necessary for understanding our main theorems, proofs, and algorithms.

\subsection{Total variation, KL divergence, and Pinsker inequality}
Here, we recall two notion of distance between probability distribution, KL divergence and total variation distance. 

\paragraph{Total variation distance (TV):}
Total variation distance between two probability distributions $X$ and $Y$ is defined as:
$$\tv(X,Y)=\sup_\calA \Pr[X\in \calA] - \Pr[Y\in \calA]$$
\paragraph{TV properties.} 
We recall briefly some properties of TV that will be useful in the remainder of this paper. 
We first note the following characterization of TV:
\begin{equation}
    \tv(X, Y) = \frac{1}{2} \sup_{f: E \to [-1, 1]} \int f dX - \int f dY .
\end{equation}

Numerous properties about $\tv$ can be derived from this characterization.
In particular, we are interested in post-processing.
Given a function $g: E \to F$, applying $g$ to the samples from $X$ and $Y$ can only decrease $\tv$:
\begin{align*}
    \tv(g(X), g(Y)) \leq \tv(X, Y).
\end{align*}
In particular if $g$ is a bijective transform, we have $\tv(g(X), g(Y)) = \tv(X, Y)$.

\paragraph{Kullback-Leibler (KL) divergence:} The KL divergence between two probability distributions $X$ and $Y$ defined over a emasurebale space $\Omega$ is defined as:

$$\kl{X}{Y}=\int_\Omega \log(\frac{dX}{dY})dX.$$

Now we are ready to state Pinsker inequality that connects total variation to KL divergence:

\paragraph{Pinsker's Inequality.} For any two probability distributions $X$ and $Y$ we have
 $$\tv(X,Y)\leq \sqrt{\frac{\kl{X}{Y}}{2}}.$$
\subsection{Differential Privacy}

\paragraph{Differential privacy.} A randomized algorithm $\calM$ satisfies $(\epsilon, \delta)$-differential privacy if, for any two datasets $D$ and $D'$ that differ in at most one sample, and for any subset $R$ of the output space, we have:
\begin{equation}
    \prob(\calM(D) \in R) \leq \exp(\epsilon) \prob(\calM(D') \in R) + \delta.
\end{equation}
While the probability of failure $\delta$ is often chosen to be inversely proportional to the number of samples, there is no consensus in the literature over desirable values of $\epsilon$.

\paragraph{Rényi differential privacy (RDP)} is a stronger definition of privacy.
For two probability distributions $P$ and $Q$  defined over $\mathcal R$, the Rényi divergence of order $\alpha > 1$ is 
\begin{equation}
    \rdp{P}{Q} \triangleq \frac{1}{\alpha - 1}\log~†\E_{x\sim Q}\left(\frac{P(x)}{Q(x)}\right)^\alpha,
\end{equation}
with $D_1$ defined by continuity $D_1 \infdivx{P}{Q} = \kl{P}{Q}$.

A randomized mechanism $f\colon \mathcal D \to \mathcal R$ satisfies $(\alpha, \epsilon)$-Rényi differential privacy (RDP) if, for any adjacent datasets $D$ and  $D'$, we have
\begin{equation}
    \rdp{f(D)}{f(D')} \leq \epsilon.
\end{equation}

Rényi divergence enjoys nice properties: it is non-decreasing in $\alpha$, and $\lim_{\alpha \to 1} \rdp{P}{Q} = \kl{P}{Q}$. 

\paragraph{DP-SGD with the subsampled Gaussian mechanism.}
DP-SGD~\citep{abadi16deep} is a modification of Stochastic Gradient Descent that makes it differentially private. 
The core mechanism behind it is the subsampled Gaussian mechanism.
Given a function $\phi$ that operates on sets $S$ with sensitivity $1$ (i.e. $ \| \phi(S) - \phi(S') \|_2 \leq 1$), the subsampled Gaussian mechanism selects sets $S$ as random and adds noise to the output  of $\phi$. 
Note that the mechanism is differentially private \emph{even} if all intermediate steps of the training process are revealed.

\paragraph{RDP accounting for DP-SGD.} \citep{abadi16deep,mironov2017renyi} propose methods to account for RDP for Gaussian mechanism. The implementations of DP-SGD \cite{Opacus} have these accounting procedures. This is important for us as we use these accounting methods to calculate one of the bounds we prove for membership inference for composition of sub-sampled Gaussian mechanims.

\subsection{Membership inference attacks}

\paragraph{Membership Inference} is the task of predicting whether a given sample was in the training set of a given model.
\citet{homer2008resolving} showed the first proof of concept, and \citet{shokri2017membership} showed that a wide variety of machine learning models are vulnerable to such attacks. 
\citet{shokri2017membership} train neural networks to attack machine learning models, and measure the success of the attack by the percentage of correctly predicted (train/test) samples, or equivalently the advantage~\citep{yeom2018privacy}. 
While \citet{shokri2017membership} trained neural networks to attack machine learning models, it was shown later that simple heuristics such as the loss~\citep{yeom2018privacy,sablayrolles2019white} is a more accurate and robust measure of membership inference.

Recent works~\citep{watson2021importance,carlini2021membership,rezaei2021difficulty} have proposed to evaluate membership inference by the precision/recall trade-off~\citep{watson2021importance} or the precision at low levels of recall~\citep{carlini2021membership}. 
In particular, such works show that some setups which were thought to be private because the membership accuracy is significantly less than $100\%$ can actually reveal membership of a small group of samples with very high precision.

There is also a line of work developed to design algorithms to specifically defend against membership inference attacks \cite{}. Although differential privacy would bound the membership inference, these empirical defenses are potentially able to achieve better utility while withstanding against existing membership inference attacks.

\section{Membership inference and total variation}
In this section, we define our security game for membership inference and show its connection to the notion of total variation distance (or statistical distance) between probability distributions.


\subsection{Security Game}

We adopt the classical assumptions of membership inference~\citep{yeom2018privacy, sablayrolles2019white, humphries2020differentially}.
We assume that data is assembled in a fixed set $D=\{z_1, \dots, z_m\}$ (resp. $D'=\{z_1, \dots, z_m, z'\}$), 
and a model $\theta$ is produced by a training algorithm $\calM$: $\theta \sim \calM(z_1, \dots, z_m)$.

Similar to \citet{humphries2020differentially}, we study in particular a more powerful adversary who knows $\theta$, $z_1, \dots, z_m$, $z'$ and wants to know whether $z'$ was used in training  $\theta$. Specifically, we use the following security game between an adversary and a challenger to measure membership inference advantage.

\begin{enumerate}
    \item Adversary picks a datasets $D=\set{z_1,\dots,z_T}$ and a data point $z'$
    \item Challenger samples a bit $b$ uniformly at random and creates $$D'=\left\{
	\begin{array}{ll}
		D \cup \set{z'}  & \mbox{if } b = 1 \\
		D & \mbox{if } b = 0
	\end{array}
\right.$$
\item Challenger learns a model $\theta$ by running $L(D')$ and sends $\theta$ to adversary.
\item Adversary guess a bit $b'$ and wins if $b'=b$.
\end{enumerate}
We then define the advantage of adversary $A$ on learning algorithm $L$ to be 
$$\adv(L,A)=2\cdot\Pr[b=b']-1.$$
We also use 

$$\adv(L)=\sup_A \adv(A,L)$$
to denote the advantage of the worst adversary against algorithm $L$.

\begin{remark}Note that we are using the notion of add/remove for neighboring datasets where the two datasets are exactly the same except that one of them has one less example. In the rest of paper, wherever we report advantage, we report it for this setting (including when we discuss the analysis of the previous works). To convert this advantage to the advantage defined for notion of neighboring datasets with \emph{replacement}, we can just double the advantage of add/remove setting. Note that doubling the advantage can potentially lead to values greater than 1, in which case the bound will be vacuous. 
\end{remark}


\begin{remark}
In section \ref{sec:membership_inference_bounds} we prove bounds on the membership advantage for composition of Gaussian mechanisms with and without sub-sampling. Both of these bounds are tight for the advantage defined based on addition/removal. However, for the notion of advantage defined based on replacement, if we double the bounds, only the bound without replacement will remain tight. We leave the question of obtaining tight bounds on the advantage based on replacement for the sub-sampled Gaussian as an open question.
\end{remark}

Let $Z$ be a random variable that corresponds to the output of the challenger in step 3 of the security game. Also let $X:= Z\mid b=0$ and $Y:=Z\mid b=1$. A deterministic adversary $A$ defines a region $\calA$ and predicts that $\theta$ is sampled from $X$ if $\theta \in \calA$ and from $Y$ if $\theta \notin \calA$. We use $X(\calA)$ and $Y(\calA)$ to denote $\Pr[X\in \calA]$ and $\Pr[Y\in \calA]$. For such an adversary we have $$\adv(L,A)=|\Pr[X\in \calA] - \Pr[Y\in \calA]|.$$

Note that with a simple averaging argument we can show that the best adversarial strategy in membership security game is a deterministic strategy. Therefore, the advantage for the learning algorithm $L$ is then defined as
\begin{align}\label{eq:adv_to_tv}
    \adv(L) = \sup_{\calA} X(\calA)-Y(\calA) 
    = \tv(X, Y),
\end{align}
where $\tv$ is the total variation distance.
Therefore, total variation distance gives us an upper bound on the advantage of any adversary. However, it is not clear how to calculate the total variation distance in general. In next subsection we discuss an approximation of total variation distance that can actually be calculated using existing techniques for RDP accounting.
\subsection{Bounding Membership Inference using Pinsker's inequality}

Using Equation \ref{eq:adv_to_tv} and directly applying Pinsker's inequality, we have:
\begin{align}
    \adv(L) \leq \tv(X,Y)\leq \sqrt{\frac{\kl{X}{Y}}{2}} = \lim_{\alpha \to 1} \sqrt{\frac{D_\alpha\infdivx{X}{Y}}{2}}
\end{align}
where $D_\alpha$ is the Renyi divergence at $\alpha$. Now one might ask why this bound is better than our bound using total variation distance. The reason we state this bound is that we have techniques for calculating the Renyi divergence of composition of adaptive and sampled Gaussian mechanisms. This enables us to calculate numerical upper bounds on the membership inference advantage of any adversary against adaptive composition of sampled Gaussian mechanisms, e.g. DP-SGD. 

To this end, we use RDP accounting to calculate the $D_\alpha$ for an $\alpha >1$. We know that for any $\alpha> 1$, $D_\alpha(X,Y)$ is greater than $\kl(X,Y)$ because $D_\alpha$ is increasing in $\alpha$. This means that for any $\alpha>1$ $D_\alpha$ will be a valid upper bound on the membership inference advantage and the bound becomes better as we decrease $\alpha$. 

In Figures \ref{fig:eps_adv} and \ref{fig:acc_adv} we calculate numerical upper bounds for membership inference for DP-SGD using typical parameters. We refer to this bound as the Pinsker bound. The figure shows that this bound obtains better numerical values compared to the bound of \citet{humphries2020differentially}. However, there are two main limitations with our Pinsker bound: 1) The bound is not tight. We are applying Pinsker's inequality which is not optimal for Gaussian mechanism.  2) It provides vacuous bounds in cases where $\kl{X}{Y}>2$. In next section, we will optimally bound the membership inference for composition of adaptive and sampled gaussian mechanisms.
\section{Membership Inference Bounds for Composition of Gaussian Mechanisms}\label{sec:membership_inference_bounds}

In this section, we show a tighter upper-bound for the adversary's advantage. 
Our main results are Theorem~\ref{thm:compgausnosub} and Theorem~\ref{thm:compgaussub}.
In particular, we will upper-bound the total variation between the \emph{transcripts} of the (subsampled) Gaussian mechanism, specifically the noisy gradients produced by the DP-SGD algorithm.
The upper bound on the result of the DP-SGD algorithm follows by application of post-processing.

The proof technique is the following: we show that the entire process of DP-SGD can be replaced by a process where each step is replaced by either $\cN(0, \sigma^2)$  or $\cN(r, \sigma^2)$ for the basic Gaussian mechanism, and replaced by $(1-q)\cN(0, \sigma^2) + q \cN(r, \sigma^2)$ for the subsampled Gaussian mechanism. 
Our proof technique relies on a modified version of $\tv$, called $\tva$.
\begin{definition}
For $a>0$ define
$\tva(X,Y)=\frac{1}{2}\int_\Omega |P(x) - a\cdot Y(x)| dx.$
\end{definition}

While we do not know an explicit form of the $\tva$ divergence between Gaussians, we can still prove that it increases with $\nrm{u_1-u_2}{2}$, as formalized in the following lemma. 
\begin{lemma}
\label{lem:monotonicity}
Let $X\equiv\cN(u_1,\sigma\cdot I_d)$ and $Y\equiv \cN(u_2,\sigma \cdot I_d)$. 
Then, for any $a\in \R^+$, $\tva(X,Y)$ is only a function of $\nrm{u_1-u_2}{2}$ and $\sigma$. 
Moreover this function is monotonically increasing with respect to $\nrm{u_1-u_2}{2}$. That is, $$\tva(X,Y)\leq \tva(\cN(0,\sigma), \cN(\nrm{u_1 - u_2}{2}, \sigma)).$$
\end{lemma}


The proof is deferred to Appendix~\ref{sec:proof_monotonicity}. 

\subsection{Warm up: Without sampling}

\begin{figure}
    \centering
    \includegraphics[width=0.4\textwidth]{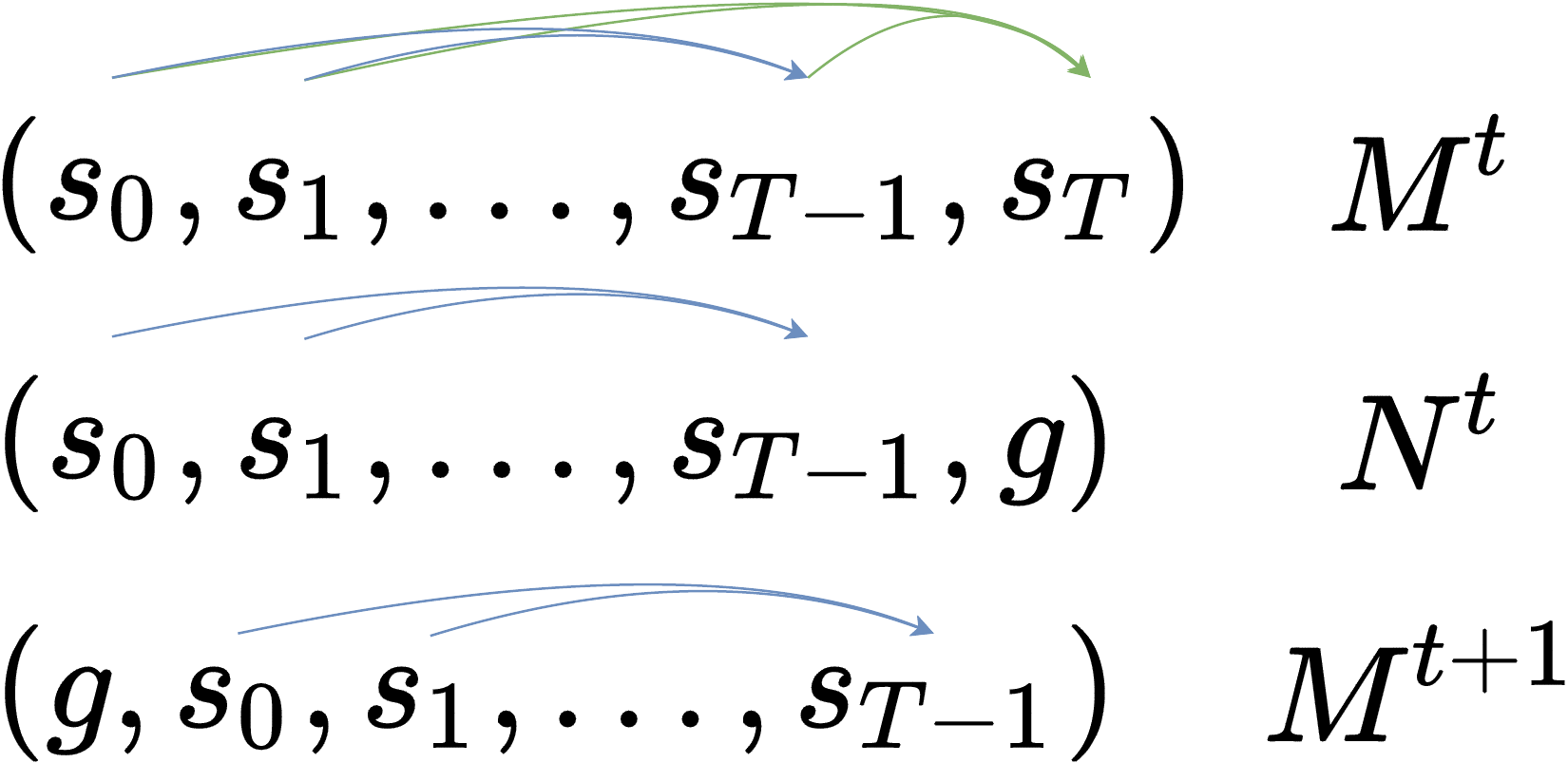}
    \caption{\label{fig:compgausnosub}
    Main component of the proof of Theorem~\ref{thm:compgausnosub}.
    The last step of DP-SGD is replaced by Gaussian noise (independent of previous steps).
    This step is then put at the start of the transcript,
    }
\end{figure}

\paragraph{Notation.} 
We use $s=(s_1,\dots, s_T)$ to denote the output (or \emph{transcript}) of a random process that consists of $T$ adaptive steps (typically the subsampled Gaussian mechanism). 
We use $\pfix{s}{t}$ to denote the first $t$ steps of the transcript of the random process.
The sampling rate $q$ is the probability of including any sample in the random set $S$. 
We use $S_i$ to denote the union of the support set of the $i$th step of the mechanism $M$ on all possible datasets. That is $S_i=\set{s_i; \exists D, s_i \in \Supp(M(D)_i)}.$ We also define $\pfix{S}{i}=S_1\times\dots,S_i$.  

\begin{theorem}[Gaussian Composition without sub-sampling]\label{thm:compgausnosub}
Let $M_1,\dots, M_T$ be a series of adaptive Gaussian Mechanisms with $L_2$ sensitivity $r$ and Gaussian noise with standard deviation $\sigma$. The membership inference risk of the composition of $M_i$'s is at most as much as a single Gaussian mechanism with sensitivity $\sqrt{T}\cdot r$ and standard deviation $\sigma$.
\begin{equation}
    \tv(X, Y) \leq \tv(\cN(0, \sigma I_T), \cN(r 1_T, \sigma I_T))
\end{equation}
\end{theorem}
\begin{proof}
Let $M^i(D)$ be the mechanism that works on $D$ (resp. $D'$) and consists of $i$ Gaussians $\cN(0, \sigma^2)$ (resp. $\cN(r, \sigma^2)$) steps followed by $T-i$ steps of DP-SGD applied to $D$ (resp. $D'$).
Specifically, the mechanism $M^0$ corresponds to DP-SGD, and $M^T$ to pure Gaussians. 
We will show in this proof that 
\begin{equation}
    \tv(M^i(D), M^i(D')) \leq \tv(M^{i+1}(D), M^{i+1}(D'))
\end{equation}
and hence 
\[\tv(X, Y) = \tv(M^0(D), M^0(D')) \leq \tv(M^T(D), M^T(D')) = \tv(\cN(0, \sigma^2 I_T), \cN(r 1_T, \sigma^2 I_T).\]

To this end, we will first argue that the very final step of the mechanism $M^i$ can be replaced with a Gaussian step without increasing the total variation distance.  
Then we can move this noise to the start of the process without affecting the result, obtaining $M^{i+1}$. 
Figure~\ref{fig:compgausnosub} illustrates this procedure. 

Let us fix a step $i$ and let $s = (s_1, \dots, s_T)$ be a transcript. 
Denoting $X \sim M^i(D)$ and $Y \sim M^i(D')$, we have


\begin{align*}
2\tv (M^i(D)&, M^i(D')) =\sum_{\pfix{s}{T}\in \pfix{S}{T}} |\Pr[X=\pfix{s}{T}] - \Pr[Y=\pfix{s}{T}]|\\
&=\sum_{\pfix{s}{T-1}\in \pfix{S}{T-1}} \sum_{s_T\in S_T}\left|\Pr[X=\pfix{s}{T}] - \Pr[Y=\pfix{s}{T}]\right|\\
&=\sum_{\pfix{s}{T-1}\in \pfix{S}{T-1}} \sum_{s_T\in S_T}\left|\Pr[X_T=s_T \mid \pfix{s}{T-1}]\cdot \Pr[\pfix{X}{T-1}=\pfix{s}{T-1}]\right. \\
& \qquad \qquad \qquad \left. - \Pr[Y_T=s_T \mid \pfix{s}{T-1}]\Pr[\pfix{Y}{T-1}=\pfix{s}{T-1}]\right|.\\
&=\sum_{\pfix{s}{T-1}\in\pfix{S}{T-1}} \Pr[\pfix{X}{T-1}=\pfix{s}{T-1}]\sum_{s_T\in S_T} \bigg| \Pr[X_T=s_T \mid \pfix{s}{T-1}] \\
&\qquad \qquad \qquad - \Pr[Y_T=s_T \mid \pfix{s}{T-1}]\frac{\Pr[\pfix{Y}{T-1}=\pfix{s}{T-1}]}{\Pr[\pfix{X}{T-1}=\pfix{s}{T-1}]} \bigg|. \numberthis \label{eq:tva_derivation}
\end{align*}
We know that the last ($T$'th) step of $M^i(D)$ (resp. $M^i(D')$) follow isotropic Gaussian distributions centered around two points $u_1$ and $u_2$ such that $\nrm{u_1-u_2}{2} \leq r$ and with standard deviation $\sigma$. 
These centers could be chosen adaptively according to the history of the mechanism. 
We use $a(\pfix{s}{T-1})$ to denote $\frac{\Pr[\pfix{Y}{T-1}=\pfix{s}{T-1}]}{\Pr[\pfix{X}{T-1}=\pfix{s}{T-1}]}$. 
By Lemma \ref{lem:monotonicity} we have 
\begin{align*} 
\sum_{s_T\in S_T}\Big|\Pr[X_T=s_T \mid \pfix{s}{T-1}] - \Pr[Y_T=s_T \mid \pfix{s}{T-1}]&a(\pfix{s}{T-1})\Big| \\
\leq~ & 2\tv_{a(\pfix{s}{T-1})}(\cN(0,\sigma^2), \cN(r,\sigma^2)).
\end{align*}

Denoting by $N^i$ the mechanism that coincides with $M^i$ for the first $T-1$ steps and is replaced by a Gaussian at the last ($T$'th) step, we thus have, by following Equation~\ref{eq:tva_derivation} in the reverse direction
\begin{align}
    2\tv(M^i(D), M^i(D')) \leq 2\tv(N^i(D), N^i(D')).
\end{align}
Given that the last step of $N^i$ does not depend on the first $T-1$ steps, we can permute to put it in first position (see Figure \ref{fig:compgausnosub}), which shows that $\tv(N^i(D), N^i(D')) = \tv(M^{i+1}(D), M^{i+1}(D'))$.
\end{proof}

\subsection{With sampling}
\paragraph{Notation.} 
We use $B(q)$ to denote a Bernoulli random variable that is equal to $1$ with probability $q$ and $0$ with probability $1-q$. 
We use $B(q)^n$ to denote an $n$ dimensional random variable where each coordinate is independent and distributed as $B(q)$. 
For $r\in \R$, we use $r\cdot B(q)^n$ to denote the random variable that is sampled by the following process: first sample from $B(q)^n$ and then multiply by $r$. 
We also  use $\cN(r\cdot B(q)^n,\sigma)$ to denote an $n$ dimensional random variable that is distributed according to mixture of Gaussians all of which have standard deviation $\sigma$ and centers are chosen at random from $r\cdot B(q)^n.$ 
\begin{theorem}[Gaussian Composition with sub-sampling]\label{thm:compgaussub}
Let $M_1,\dots, M_T$ be a series of adaptive Gaussian Mechanisms with $L_2$ sensitivity $r$ and Gaussian noise with standard deviation $\sigma$ and sub-sampling rate $q$. The membership inference risk of the composition of $M_i$'s is at most
$$ \tv\big(\cN(0, \sigma), \cN(r\cdot B(q)^n,\sigma)\big)$$
\end{theorem}


\begin{lemma}\label{sub-sampling_tva}
Let $X'\equiv (1-q)\cdot Y + q\cdot X$ then we have
$$\tva(X',Y)= q\tv_{\frac{a+q-1}{q}}(X,Y)$$
\end{lemma}
\begin{proof}
We have 
\begin{align*}
    2\tva(X',Y)= \int_{\Omega}\left|X'(x)-aY(x)\right|dx 
    &= \int_{\Omega}\left|qX(x)-(q+a-1)Y(x)\right|dx\\
    &=q\int_{\Omega}\left|X(x)-\frac{(q+a-1)}{q}Y(x)\right|dx \\
    &=2q\tv_{\frac{a+q-1}{q}}(X,Y).
\end{align*}
\end{proof}

\begin{proof}[Proof of Theorem \ref{thm:compgaussub}]
The proof steps are similar to that of Theorem \ref{thm:compgausnosub}. First, we have

\begin{align*}
2\tv(X,Y) 
&=\sum_{\pfix{s}{T-1}\in \pfix{S}{T-1}} \Pr[\pfix{X}{T-1}=\pfix{s}{T-1}] \cdot\\
& \qquad \left(\sum_{s_T}\Big|\Pr[X_T=s_T \mid \pfix{s}{T-1}]- \Pr[Y_T=s_T \mid \pfix{s}{T-1}]\frac{\Pr[\pfix{Y}{T-1}=\pfix{s}{T-1}]}{\Pr[\pfix{X}{T-1}=\pfix{s}{T-1}]}\Big|\right)\\
&=2\sum_{\pfix{s}{T-1}\in \pfix{S}{T-1}} \Pr[\pfix{X}{T-1}=\pfix{s}{T-1}] \tv_{a(\pfix{s}{T-1})}(X_T\mid \pfix{s}{T-1},Y_T\mid \pfix{s}{T-1}) .
\end{align*}
But since $X_T$ and $Y_T$ are subsampled Gaussian mechanisms we have $X_T\equiv (1-q)Y_T + qX'_T$ where $Y$ and $X'$ are mixtures of Gaussians. Therefore, by Lemma \ref{lem:monotonicity} and Lemma \ref{sub-sampling_tva} we have
\begin{align*}
&\tv(X,Y) \\
&~~= \sum_{\pfix{s}{T-1}\in \pfix{S}{T-1}} \Pr[\pfix{X}{T-1}=\pfix{s}{T-1}\in \pfix{S}{T-1}] q\tv_{\frac{a(\pfix{s}{T-1})+q-1}{q}}(X'_T\mid \pfix{s}{T-1}, Y_T\mid \pfix{s}{T-1})\text{~~(By Lemma \ref{sub-sampling_tva})}\\
&~~\leq \sum_{\pfix{s}{T-1}\in \pfix{S}{T-1}} \Pr[\pfix{X}{T-1}=\pfix{s}{T-1}\in \pfix{S}{T-1}] q\tv_{\frac{a(\pfix{s}{T-1})+q-1}{q}}(\cN(0,\sigma),\cN(r,\sigma))\text{~~(By Lemma \ref{lem:monotonicity})}\\
&~~= \sum_{\pfix{s}{T-1}\in \pfix{S}{T-1}} \Pr[\pfix{X}{T-1}=\pfix{s}{T-1}] \tv_{a(\pfix{s}{T-1})}(\cN(0,\sigma),(1-q)\cN(0,\sigma) + q\cN(r,\sigma))\text{~~(By Lemma \ref{sub-sampling_tva})}\\
&~~= \sum_{\pfix{s}{T-1}\in \pfix{S}{T-1}} \Pr[\pfix{X}{T-1}=\pfix{s}{T-1}] \tv_{a(\pfix{s}{T-1})}(\cN(0,\sigma),\cN(r\cdot B(q),\sigma) ).
\end{align*}


Therefore, we can replace $X_T$ with a mixture of two Gaussians centered at $0$ and $r$ and $Y_T$ with a single Gaussian centered at $0.$ Now we can use the same technique used in proof of Theorem \ref{thm:compgausnosub} and move $X_T$ and $Y_T$ to the first round and repeat this process. 
At the end, $Y$ will turn to a $n$-dimensional Gaussian centered at $0$ and standard deviation $\sigma$ and  $X$ will be a mixture of Gaussians with center randomly selected according to a $n$-dimensional Bernoulli distribution with probability $q$. 
That is, the advantage is bounded by 
$$\tv(\cN(0^n, \sigma),\cN(rB(q)^n,\sigma))$$
\end{proof}

\begin{remark}
Theorem \ref{thm:compgausnosub} could be simply extended to composition of Gaussian mechanisms with varying noise levels. However, if the noise levels are selected adaptively, the proof is not clear. We leave the composition of Gaussians with adaptive noise selection as an open question.
\end{remark}
\begin{remark}
Although Theorem \ref{thm:compgausnosub} is stated only for Gaussian mechanism, the Theorem extends to any mechanism that satisfies monotonicity under $\tv_a$ according to some notion of sensitivity. For example, if one can show that $\tv_a(\mathcal{L}(0,\sigma), \mathcal{L}(u,\sigma))$ is monotonically increasing with respect to $|u|_1$, then Theorem \ref{thm:compgausnosub} extends to composition of Laplace mechanisms with bounded $\ell_1$ sensitivity.
\end{remark}

\subsection{Numerical computation}

In order to numerically approximate the upper-bound, we first convert the notion of $\tv$ into a expectation formulation as follows:

\begin{align}
    \tv(X, Y) &= \int_\Omega \left(X(t) - Y(t)\right) \mathbbm{1}\left(Y(t) \leq X(t)\right) dt \\
    &= \int_\Omega \left(1 - \frac{Y(t)}{X(t)}\right) \mathbbm{1}\left(Y(t) \leq X(t)\right) X(t) dt \\
    &= \E_{t \sim X} \left( \left(1 - \frac{Y(t)}{X(t)}\right) \mathbbm{1}\left(Y(t) \leq X(t)\right) \right)
\end{align}

Note that this expectation is over distribution $X$. So we can sample a dataset from $X$ and approximate this expectation using empirical averaging (or Monte-Carlo sampling):
\begin{align}
    \tv(X, Y) \approx \frac{1}{m} \sum_{i=1}^m  \left(1 - \frac{Y(t_i)}{X(t_i)}\right) \mathbbm{1}\left(Y(t_i) \leq X(t_i)\right)
\end{align}
We know that Monte-Carlo estimation of this expectation using is very precise because the quantity is bounded between $0$ and $1$. 

Note that in-order to calculate this, we need to calculate $Y(t_i)/X(t_i)$ and that is possible because we have the mathematical form of the probability distribution function of for $Y$ and $X$. In Figures \ref{fig:eps_adv} and \ref{fig:acc_adv} we calculate the upper bound using this Monte-Carlo simulation for typical settings in DP-SGD.

\section{Conclusion}
In this paper, we directly analyzed membership inference bounds for composition of adaptive sampled Gaussian mechanisms. Our analysis enables us to obtain bounds that are much better that one can obtain by converting differential privacy guarantees to membership inference guarantees. Our analysis shows that although differential privacy guarantees might sometimes large membership inference guarantees, but the mechanisms that obtain differential privacy can be in fact much more secure against membership inference attacks. Previously, this phenomenon was observed for DP-SGD and here for the first time we prove it. 

Our analysis is the first to directly analyze membersihp inference bounds. We limited our study to membership inference attacks against sampled Gaussian mechanisms as DP-SGD is the most used differential private learning algorithm. But this kind of membership inference analysis could be potentially done for other mechanisms and algorithm. We leave this for future work.

We also note that The parameters in DP-SGD that achieve optimal membership privacy versus utility might be different than that of differential privacy. Our new analysis opens up the possibility of a systematic search for optimal hyper parameters to obtain optimal utility for a given upper bound on membership inference advantage.


\clearpage
\bibliography{main}
\bibliographystyle{plainnat}

\clearpage
\appendix

\section{Proof of Lemma~\ref{lem:monotonicity}}
\label{sec:proof_monotonicity}


\begin{proof}
The first part follows by the symmetry of isotropic Gaussian. For the second part (monotonicity) we use the definition of  $\tva$. Without loss of generality we can assume $a\in [0,1]$ as otherwise we can work with $\tva(P,Q)/a = \tv_{1/a}(Q,P)$. 
Let $r=\nrm{u_1-u_2}{2}$. We can show that the derivative of the integral is always positive. In the following calculations, we use $c_1, c_2, c_3$ and $c_4$ to denote positive constants that are independent of $r$.
\newcommand{\midpoint}{\frac{r^2-2\sigma^2\ln(a)}{2r}}
\newcommand{\midpointsigma}{\frac{r^2-\ln(a)\sigma^2}{2\sqrt{2}\sigma r}}
\newcommand{\midpointsigmar}{\frac{-r^2-\ln(a)\sigma^2}{2\sqrt{2}\sigma r}}

First note that $x^*=\midpoint$ is a middle point where $e^{-\frac{x^2}{2\sigma^2}} - ae^{{-\frac{(x-r)^2}{2\sigma^2}}}$ goes from positive to negative as $x$ increases. 
By our assumption that $a \in [0, 1]$, we have that $x^* > 0$. 
Recalling that $\erf(z) = \frac{2}{\sqrt{\pi}} \int_0^{z} \exp(-t^2) dt$, and that
$\erf(\infty) = 1$ so that (by symmetry) $\frac{2}{\sqrt{\pi}} \int_{-\infty}^{0} \exp(-t^2) dt = 1$, we can write
\begin{align*}
    \tva(P,Q)&= c_1\left(\int_{-\infty}^{\infty} \left|e^{-\frac{x^2}{2\sigma^2}} - ae^{{-\frac{(x-r)^2}{2\sigma^2}}}\right|dx\right)\\
    &=c_1\left(\int_{-\infty}^{x^*}e^{-\frac{x^2}{2\sigma^2}} - ae^{-\frac{(x-r)^2}{2\sigma^2}} + \int_{x^*}^{\infty}     ae^{-\frac{(x-r)^2}{2\sigma^2}} - e^{-\frac{x^2}{2\sigma^2}}\right)\\
    & = c_1 \left( 1 + \erf\left(x^*/\sqrt{2}\sigma\right) - a \erf( (x^* - r)/\sqrt{2}{\sigma}) \right) \\
    &~~+ \left(a(1 - \erf\left((x^* - r)/\sqrt{2}{\sigma}\right) + (1 - \erf\left(x^*/\sqrt{2}\sigma\right) \right) \\
    &=c_2 \left(\erf\left(\midpointsigma\right)+ 1 - a\erf\left(\midpointsigmar \right) - a\right) .
\end{align*}
Now, let $f_1(r) =\erf\left(\midpointsigma\right)$ and $f_2(r)=- a\erf\left(\midpointsigmar \right)$.
Taking the derivative with respect to $r$ we have

$$\frac{\partial f_1}{\partial r} = c_3\left(\frac{1}{2\sqrt{2}\sigma}+\frac{\ln(a)\sigma}{2\sqrt{2}r^2}\right)e^{-\left(\midpointsigma\right)^2}$$
$$\frac{\partial f_2}{\partial r} = c_3 a \left(\frac{1}{2\sqrt{2}\sigma}-\frac{\ln(a)\sigma}{2\sqrt{2}r^2}\right)e^{-\left(\midpointsigmar\right)^2}$$
\newcommand{\var}{e^{\frac{-r^4 - \ln(a)^2\sigma^4}{8\sigma^2r^2}}}
Now note that we have $e^{-\left(\midpointsigma\right)^2} = a^{1/2}\cdot e^{-\left(\midpointsigmar\right)^2}$. 
Therefore, we have
$$c_4\frac{\partial \tva }{\partial r} = e^{-\left(\midpointsigmar\right)^2}\cdot\left(\frac{1+\sqrt{a}}{2\sqrt{2}\sigma} + \frac{\ln(a)\left(\sqrt{a}-1\right)\sigma}{2\sqrt{2}r^2}\right).$$
Now since $a\in [0,1]$, we have $\ln(a)\leq 0$ and $\sqrt{a}-1<0$. Which means the term $\frac{1+\sqrt{a}}{2\sqrt{2}\sigma} + \frac{\ln(a)(\sqrt{a}-1)\sigma}{2\sqrt{2}r^2}$ is positive. This implies that the whole gradient is positive.

\end{proof}

\section{Membership inference precision}

In this section, we refine the analysis of \citet{sablayrolles2019white} for the accuracy of a membership attack.

\paragraph{Upper-bound on precision.} 
Let us first derive a bound on the precision of membership inference.
We assume that there are two datasets $D$ and $D'$ and that a differentially-private mechanism $\calM$ trains a model represented by $\theta$. 

With probability $(1 - \delta$) over the choice of $\theta$, we have:
\begin{align}
    -\epsilon \leq \log \left( \frac{\prob(M(D)=\theta)}{\prob(M(D')=\theta)} \right) \leq \epsilon
\end{align}

Given that there is a balanced prior $\prob(D) = \prob(D')$, using Bayes rule, we have:
\begin{align}
    \prob(D~|~\theta) &= \frac{\prob(M(D)=\theta) \prob(D)}{\prob(M(D)=\theta) \prob(D) + \prob(M(D')=\theta) \prob(D')} \\
    &= \frac{\prob(M(D)=\theta)}{\prob(M(D)=\theta) + \prob(M(D')=\theta) } \\
    &= \sigma \left( \log \left( \frac{\prob(M(D)=\theta)}{\prob(M(D')=\theta)} \right) \right),
\end{align}
with $\sigma(u) = 1 / (1 + \exp(-u))$ the sigmoid function. 

Hence the precision $\prob(D~|~\theta)$ is bounded between $\sigma(-\epsilon)$ and $\sigma(\epsilon)$, as $\sigma(\cdot)$ is non decreasing.

\paragraph{Upper-bound on attack accuracy.}

The accuracy of the Bayes classifier is 
\begin{align}
    \operatorname{Acc} = \max(\prob(D~|~\theta), 1 - \prob(D~|~\theta)),
\end{align}

and thus
\begin{align}
    \operatorname{Acc} &\leq \max(\sigma(\epsilon), \sigma(-\epsilon)) \\
    &= \sigma(\epsilon)
\end{align}

This means that the attack accuracy is bounded by $\sigma(\epsilon)$ with probability $1 - \delta$.
Empirically, we see that the sigmoid function closely matches the bound given by \citet{humphries2020differentially}.
Simply stated, this derivation shows that the bound proven by \citet{humphries2020differentially} actually holds with probability $ 1 - \delta$ instead of on average.

\end{document}